\documentclass[twoside,leqno,twocolumn]{article}
\usepackage{wrapfig}
\usepackage[affil-it]{authblk}
\usepackage{amsmath, amsfonts, amssymb, graphicx, float}%, upgreek}
\usepackage{amsbsy}
\usepackage{mathrsfs,multirow}
\usepackage{bm,comment}
\usepackage{verbatim}
\usepackage{caption,subcaption}
\usepackage[bookmarks]{hyperref}
\usepackage{algorithm}
\usepackage{url,color}
\usepackage{latexsym}
\usepackage{dblfloatfix}
\newcommand{\secref}[1]{Section~\ref{#1}}
\newcommand{\figref}[1]{Figure~\ref{#1}}

\newcommand{\algref}[1]{Algorithm~\ref{#1}}
\newcommand{\remove}[1]{}

\newtheorem{theorem}{Theorem}
\newtheorem{problem}{Problem}
\newtheorem{formulation}{Formulation}
\newtheorem{conjecture}{Conjecture}
\newtheorem{lemma}{Lemma}
\newtheorem{definition}{Definition}
\newtheorem{corollary}{Corollary}
\newenvironment{proof}[1][Proof]{\begin{trivlist}
\item[\hskip \labelsep {\bfseries #1}]}{\end{trivlist}}

\def\DD{\boldsymbol{D}}

\def\II{\boldsymbol{I}}
\def\LL{\boldsymbol{\mathcal{L}}}
\def\AAcal{\boldsymbol{\mathcal{A}}}
\def\TT{\boldsymbol{T}}
\def\XX{\boldsymbol{X}}

\def\WW{\boldsymbol{W}}
\def\AA{\boldsymbol{A}}
\def\MM{\boldsymbol{M}}

\def\ppi{\boldsymbol{\pi}}
\def\PPi{\boldsymbol{\Pi}}

\def\GGamma{\boldsymbol{\Gamma}}

\def\PPP{\mathbb{P}}
\def\WWW{\mathbb{W}}

\def\ppi{\boldsymbol{\pi}}

\title{Network Composition from Multi-layer Data}
\author{Kristina Lerman}
\affil{Information Sciences Institute, University of Southern California}
\author{Shang-Hua Teng}
\affil{Computer Science, University of Southern California}
\author{Xiaoran Yan}
\affil{Network Science Institute, Indiana University}

\begin{document}
\date{\vspace{-5ex}}
\maketitle

\begin{abstract}
%To model complex social-technological systems, we propose a mathematical framework for building networks from multiple layers. Unlike previous approaches, we take into account the different dynamical proccesses unfolding within and across the layers. Based on the idea that both topology and dynamics determines the perceived structures, the framework will properly transform each layer and build inter-layer structures that explicitly models a unified dynamical process. This construction allows us to define consistent measures for vertex centrality and community structure across the layers. We will demonstrate that the new framework can help us better identify important or vulnerable components in multi-layer networks with complex dynamics.

It is common for people to access multiple social networks,  
  for example, using phone, email, and social media.
Together, the multi-layer social interactions form a 
  ``integrated social network.''
% with network layer potentially representing a different set of 
How can we extend well developed knowledge about single-layer networks, 
including vertex centrality and community structure, to such heterogeneous structures?
In this paper, we approach these challenges by proposing a principled framework of {\em
  network composition} based on a unified dynamical process. 
Mathematically, we consider the following abstract problem: 
Given {\em multi-layer network data}, $(G^1,\ldots,G^l)$ over a 
vertex set $V$ and additional parameters for {\em intra and inter-layer dynamics},
construct a (single) weighted network $G$ that best integrates the joint process. We use transformations of dynamics to unify heterogeneous layers under a common dynamics. For inter-layer compositions, we will consider several cases as the inter-layer dynamics plays different roles in various social or technological networks. Empirically, we provide examples to highlight the usefulness of this framework for network analysis and network design. 
\end{abstract}

\section{Introduction}
\label{sec:introduction}
As a powerful representation for many complex systems, networks model entities and their interactions as vertices and edges. Studies of network structures, including those of \emph{vertex centrality} and \emph{community structure} have lead to fundamental insights into the organization and function of social, biological and technological systems~\cite{newman2010networks, bonacich1987power, Fortunato10}. On top of these network structures, different dynamical processes unfold~\cite{Borgatti05, ghosh_rethinking_2012, Ghosh2014KDD}. Our ability to model and predict dynamic network phenomena has led to new applications ranging from ranking web pages to maximizing social influence and controlling epidemics~\cite{Page99thepagerank, Colizza2006, Kempe03}.
%Previous work also revealed the mathematical interplay between network structures and dynamics \cite{Borgatti05, ghosh_rethinking_2012, Ghosh2014KDD}.

Traditionally, most research has focused on the simple graph representation where all verticies and edges are of a single type. More recently, there has been great interest in going beyond such a homogeneous model to investigate networks that are capable of capturing multiple types of connections. Extensive efforts towards such heterogeneous models came from both social~\cite{verbrugge1979multiplexity, wasserman1994social} and computational disciplines~\cite{kolda2009tensor, acar2009unsupervised}, as the simple graph abstraction are often too crude a description of reality. For example, it is very common for people to have interactions across multiple social networks, including neighbors, coworkers, and also online interactions through email and social platforms, such as Facebook and Twitter. Each of these networks underlies a different type of social interactions. Because of the different origins and motivations, many names have been given to such heterogeneous models including but not limited to multiplex 
network, multi-
relational networks, and networks of networks. For a comprehensive review see~\cite{kivela_multilayer_2013}. In this paper, we adopt their terminology and use the general model of {\em multi-layer networks}, with {\em multiplex network} being a special case when inter-layer structures are absent.

Structure and dynamics of multi-layer networks have been explored in both theoretical graphs and real world data~\cite{balcan_multiscale_2009, buldyrev_catastrophic_2010, de_domenico_random_2013, gomez_diffusion_2013,gallotti_multilayer_2015}, with several researchers generalizing community structure and centrality measures to multiple edge types~\cite{mucha_community_2010, Michoel2012spectral,hu_multislice_2012,bazzi_community_2015,sole-ribalta_centrality_2014, taylor_eigenvector-based_2015}.
%Some explored the structure-dynamics interplay by introducing dynamics-aware measures~\cite{de_domenico_centrality_2013}.
%
However, it remains an open research question as how to build a multi-layer network in the first place. Such networks are often constructed simply by stacking or projecting layers into a single network. When inter-layer edges are explicitly modeled, they usually appear as tunable parameters, despite the fact that general theoretical framework allows much richer representations~\cite{kivela_multilayer_2013}. One challenge for modeling inter-layer structures is that they are empirically difficult to measure in most cases~\cite{gallotti_multilayer_2015,gallotti_information_2015}.

In this paper, instead of proposing multi-layer generalizations of 
  network measures, we
  approach the problem by constructing 
  a single composed network which integrates the multi-layer data.
We propose a two-stage framework for multi-layer network composition based on a unified dynamical process, as illustrated by \figref{fig:framework}. Specifically, the first stage address the layer heterogeneity through layer transformations. In \secref{sec:layer}, we discuss how to transform the layers into homogeneous Markov processes using the framework of the parameterized Laplacian~\cite{Ghosh2014KDD}. In \secref{sec:inter}, we will discuss the second stage, which consists of several ways of combining layers in commonly seen social and technological networks. They include how to construct multiplex networks as well as multi-layer structures with observed inter-layer dynamics. In the multi-layer case, we choose to model a vertex's inter-layer transitions, i.e., its participation in the different network layers, as another Markov process. Therefore, we can treat the combined layers of interlinked dynamical processes as a joint Markov process itself.
We will prove that under this view a unique composition of a multi-layer network exists. In practice, however, it is difficult to fully observe inter-layer transitions. Hence, we will also consider the problem of network composition with partial information, for example, knowing only the stationary distribution of a inter-layer Markov process. 
%In addition, our formulation assumes Markovian dynamics. Clearly, not all dynamics on networks are Markovian in nature. 
Together, these dynamical process based transformations and compositions capture the heterogeneous structure while leaving a unified underlying topology. As a result, we can directly apply existing network algorithms of vertex centrality and community detection to the correctly composed joint structure. Some of the applications of multi-layer formalism to real-world data are discussed in \secref{sec:exp}.

\begin{figure}
    \centering
    \includegraphics[width=0.5\textwidth]{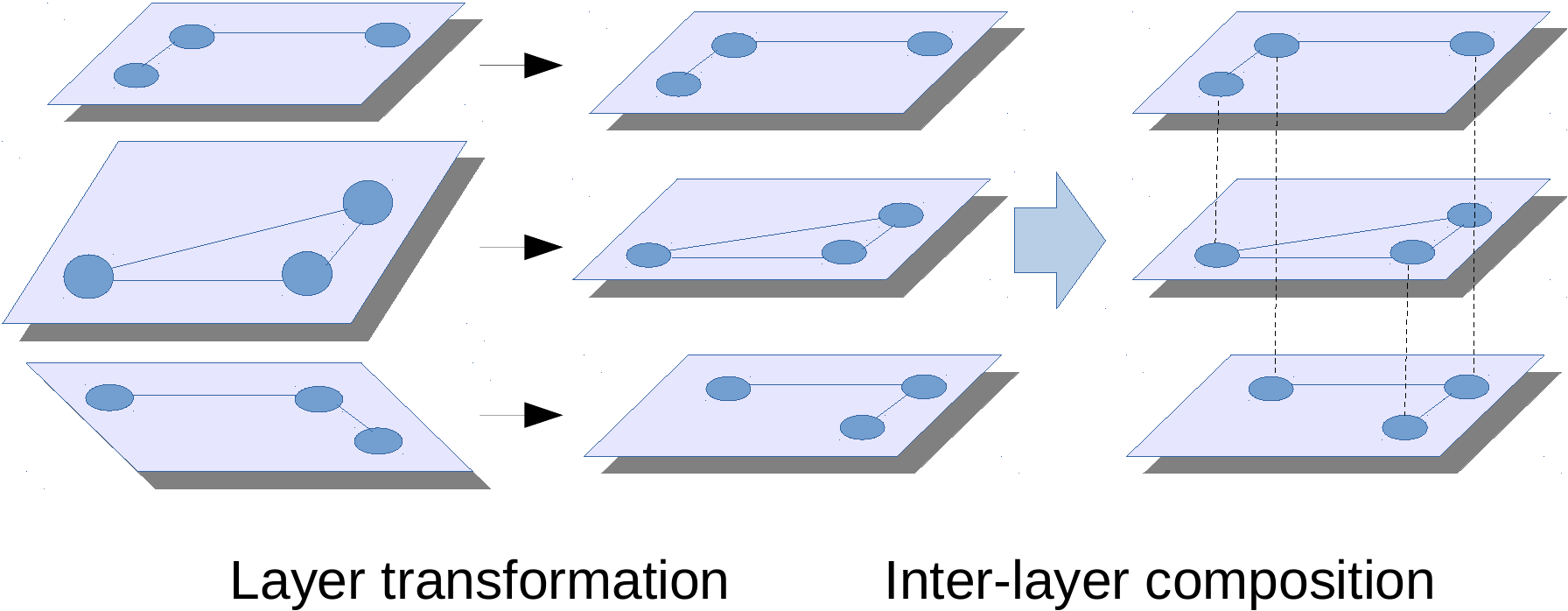}
    \caption{A two-stage framework for multi-layer composition}
    \label{fig:framework}
\end{figure}

\section{Preliminaries}
\def\QQ{\boldsymbol{Q}}

\label{sec:preliminary}
In this section, 
  we introduce some basic notations 
  for multi-layer networks and dynamical processes.
%% In particular, we will review the relation 
%%   between the transition matrix of a Markov model 
%%   and the adjacency matrices of networks
\vspace*{0.05in}

\noindent {\em Single-layer data}:
A standard network is represented by weighted directed
 graph $G = (V,E,\AA)$, where 
 $V = \{1,...,n\}$ and for $u,v\in V$, 
 $a_{uv}\geq 0$ assigns an affinity weight to edge $(u,v)\in E$.
We follow the convention that $a_{uv} = 0$ if and only if 
  $(u,v)\not\in E$. 
$G$ may have self-loops, and edges in $G$ are assumed to be 
  directed.
In other words, the weighted adjacency matrix $\AA$ can be asymmetric
  and can have none-zero entries on the diagonal.
For $u \in V$, let $d^{out}_u =  \sum_{v=1}^n a_{u,v}$ denote the {\em
  out-degree}
  of vertex $u$.
Similarly, let $d^{in}_u =  \sum_{v=1}^n a_{v,u}$ denote the 
 {\em in-degree}  of vertex $u$.
In this paper, we use $\DD_{\AA}$ (or $\DD$ when the context is clear)
   to denote the diagonal matrix whose entries 
  are out-degrees.
\vspace*{0.05in}

\noindent {\em Multi-layer data}:
We consider {\em vertex-aligned} multi-layer networks
% are called  % according to the definitions in
\cite{kivela_multilayer_2013}. 
We usually use $l$ to denote number of layers,
  and use $G^i= (V,E^i,\AA^i)$ to denote the network
  at $i^{th}$ layer.
For clarity, we will use superscripts $i,j,r$ 
  for the layers and subscripts $u,v,w$ for vertices. 
Note that the vertex set $V$ is the same across the layers.
\figref{fig:example} is a toy example
  of a three-layer network,
% with three horizontal layers, 
  consisting of (hypothetical) phone contacts, email exchanges
  and Facebook friendships of four users.
In the figure, users appear in multiple layers, connected 
  by a dashed line.
% If a vertex is missing in a layer, vertex alignment can still be achieved by padding isolated vertices.
\begin{figure}
\begin{center}
  \begin{subfigure}[b]{0.38\textwidth}
    \includegraphics[width=\textwidth]{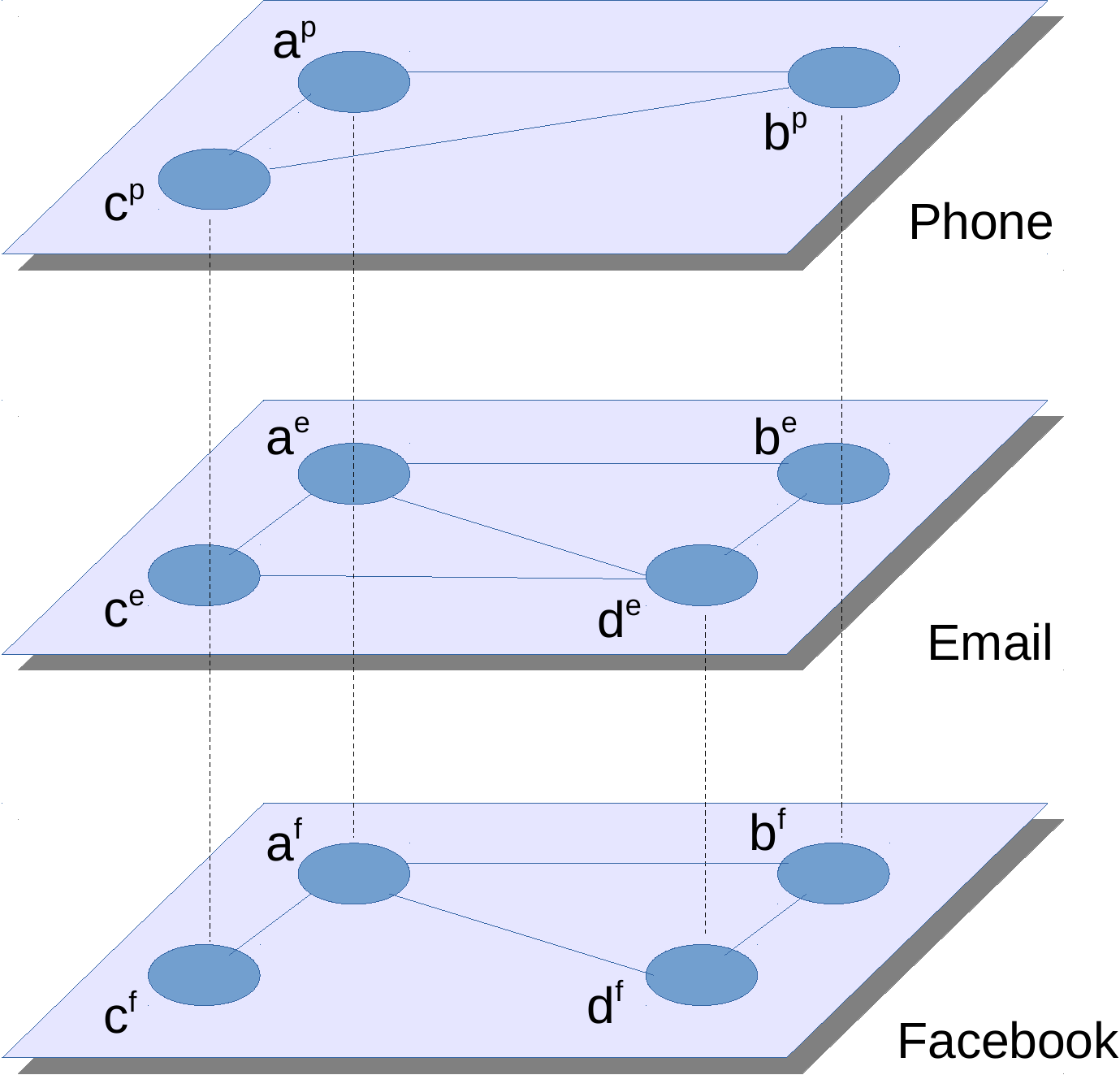}
%    \caption{}
  \end{subfigure}
%  \begin{subfigure}[b]{0.28\textwidth}
%    \includegraphics[width=\textwidth]{../figures/vertical-crop}
%    \caption{}
%  \end{subfigure}
\end{center}
\caption{A toy example of multi-layer social network in horizontal perspective}
\label{fig:example}
\end{figure} 

In this paper we take a dynamical view of the network structure. The simplest dynamical process on graphs $G$ is the discrete time \emph{unbiased random walk} (URW), represented by the transition matrix $\MM$. For their connections we have the following lemma:

\begin{lemma}
\label{the:mapping}
For every directed network $G =\\ (V,E,\AA)$, 
  there is a unique transition matrix, $\MM_{\AA} = \AA\DD_{\AA}^{-1} $,
  that captures the URW Markov process on $G$.
Conversely, given a transition matrix $\MM$, there is in fact 
  an infinite family of adjacency matrices whose random walk Markov process
  is  consistent with $\MM$: 
\[
\AAcal_{\MM} = \{\MM\GGamma : \mbox{$\GGamma$ is a positive diagonal matrix}.\}
\]
\end{lemma}

In other words, every 
  directed network uniquely defines a random walk
  process. However, given a transition matrix $\MM$,
   there remains $n$ degrees of freedom
   to specify the underlying network. Intuitively, they are vertex scaling factors bacause each random walk distribution remains the same as along as the whole column is multiplied together.
We will use this fact in some of our construction.

%% \begin{lemma}
%% Suppose $\AA_1$  and $\AA_2$ are two $n \times n$ weighted adjacency matrices.
%% Then, $\MM_{\AA_1} = \MM_{\AA_2}$ if and only there exists a diagonal matrix 
%% $\GGamma$, such that $\GGamma \AA_1 = \AA_2$.
%% \end{lemma}

Recall that an $n$-dimensional probability 
 vector $\ppi$ is the {\em stationary distribution}
  of $\MM$ if $\MM\ppi = \ppi$.
The Markov process defined by $\MM$ 
  is {\em detailed-balanced}
  if for all $u,v\in V$, $\pi_u m_{uv} = \pi_v m_{vu}$.
It is well known that \cite{aldous2002reversible}:
\begin{lemma}
\label{the:mappingUn}
Suppose $\MM$ is a detailed-balanced 
  transition matrix with stationary $\ppi$.
Let $\PPi$ be the diagonal matrix defined by $\ppi$. 
Then, for all $\alpha > 0$, $\alpha\cdot \MM\PPi$ is symmetric.
Namely, $\AAcal_{\MM} = \{\alpha\cdot \MM\PPi\}$ contains symmetric adjacency matrices
  if and only if $\MM$ is detailed-balanced.
\end{lemma}

Therefore, for undirected graphs, there is only one degree of freedom to specify the underlying network for a given detailed-balanced transition matrix, which can be interpreted as the global scaling factor.

\section{Layer Transformation}
\label{sec:layer}
In \cite{Ghosh2014KDD}, Ghosh et al. argued that perceived network structure is a result of the interplay between the network topology and the dynamical process on top of it. We believe this interplay is even more pronounced in multilayer networks, with each layer represents a different type of connection. It is essential to account for the different intra-layer dynamics before the composition.

Taking \figref{fig:example} for example, if we want to trace a message in the combined network, one should take into account the different propagating patterns in each layer. We might weigh the edges in the phone layer much heavier if it is a business message. Similarly, each user may also has its own habits in terms of how often they check their email and Facebook accounts.

In \secref{sec:preliminary}, we showed the mathematical mapping from the adjacency matrices to the transition matrices representing simple URWs. The parametrized Laplacian operator introduced in \cite{Ghosh2014KDD} can model a richer family of dynamical processes. As a conservative operator, it has a dominate eigenvalue $0$, and models continuous time random random walks or consensus processes with various biases and delays at vertices. In this paper, we shall focus on the random walk formulation:
\begin{equation}
\label{eq:parametrizedL}
\LL=(\DD'-B\AA )(\DD'\TT)^{-1}.
\end{equation}
where $\AA$ represents the adjacency matrix of the and $\DD'$ is the reweighed diagonal degree matrix \footnote{The general solution of the continuous time model is $\theta(t) = e^{-\LL t} \theta(0)$, although $P = \II-\LL$ can also be interpreted as the stochastic matrix of a discrete random walk.}.

Compared with traditional Laplacians, the parametrized Laplacian has two additional parameters: $\TT$ and $B$. The diagonal matrix $\TT$ controls the time delay factors, or local clock rate, at each vertex. In the toy example \figref{fig:example}, $\TT$ can be used to capture the checking frequency of email and Facebook accounts, or the limited attention a user has in face of information overload \cite{Hodas12limited}. It models user activities as Poison processes where the waiting time between logins are exponentially distributed with means specified by $\TT$ \cite{lambiotte_laplacian_2008}. Without loss of generality, we constrain all the entries in $\TT$ with $\tau_u \geq 1$. \footnote{Others have argued for more realistic models such as \cite{barabasi_origin_2005}, we stick to Poison processes for its mathematical simplicity.}

The bias factors form the other diagonal matrix $B$. It changes the trajectory by giving random walk targets different weights. Note that the degree matrix $\DD'$ is now defined as: $d'_u= \sum_v [B\AA] _{uv}$. In \figref{fig:example}, $B$ can be used to model different routing strategies in each layer. Such routing biases can based on structural properties like vertex degree or some external attributes specified by $B$. Entries $b_i$ of $B$ can be quite general, as long as the entries of $B\AA$ remains non-negative.

\begin{comment}
We can specify the more general network composition problem as:
\begin{problem}
\label{prob:general}
{\textbf{Network composition with layer transformations}}\\
\noindent Given weighted network layers: $G^1= (V, E^1,\AA^1), G^2= (V, E^2, \AA^2), ..., G^l= (V, E^l, \AA^l)$, and the parameters of the dynamics $\TT^1,B^1,\TT^2,B^2,..., \TT^l, B^l$, together with $l\times l$ inter-layer Laplacian matrices $M_u$ for each vertex $u\in V$,
Compose a super adjacency matrix $\mathbb{W}$ to integrate the multilayer network data.
\end{problem}

To ensure consistent compositions, we define 
\begin{definition}(Parametrized Markovian Consistency)
We call $\mathbb{W}$ is Markovian consistent with the Parametrized Laplacian inputs, if its random walk model $\PPP$ satisfies 
\begin{enumerate}
\item can be ``layerly projected''  random-walk Markov model
  corresponding to every parametrized Laplacian layer, and 
\item can be ``consistently individualized'' into all personalized 
      inter-layer dynamics given by $(\MM_u :u\in V)$.      
\end {enumerate}
\end{definition}
\end{comment}

To reach homogeneity across the layers, we need to transform each input layer to equivalent graphs with URWs as the unifying dynamics. Using the two graph transformations under the parametrized Laplacian framework, we have

\begin{theorem}
\label{theorem:transform}
For a directed network $G =\\ (V,E,\AA)$, the dynamics $\LL=(\DD'-B\AA)(\DD'\TT)^{-1}$ is equivalent to a URW on another transformed graph.
\end{theorem}

\begin{proof}
The first transformation under the parametrized Laplacian framework is the bias transformation, which is already defined in the parentheses of Equation \eqref{eq:parametrizedL}. A biased random walk from vertex $u$ to $v$ with transition probability $P_{vu} \propto b_v a_{vu}$ is equivalent to an unbiased random walk on a reweighed adjacency matrix:
$\AA' = B\AA$, because
$$P'_{vu} = \dfrac{a'_{vu}}{\sum_{v}a'_{vu}} \propto b_v a_{vu} \propto P_{vu}\;.
$$
The second transformation is to view the delay factors $\TT$ as self-loops. Under the parametrized Laplacian framework, delays can be understood as rescaling the mean waiting time of random walks at vertices. They can be absorbed in to the scaled adjacency matrix $\WW$, which we call the interaction matrix. On top of a bias transformed random walk adjacency $\AA'$, we apply the delay factors:
\begin{align*}
 &(\DD'-\AA')\DD'^{-1} \TT^{-1}= \II\TT^{-1} - \AA'\DD'^{-1}\TT^{-1}\\
				    =& \II - (\II- \TT^{-1}) - \AA'\DD'^{-1}\TT^{-1}\\
				    =& \II - (\TT- \II + \AA'\DD'^{-1})\TT^{-1}\\
				    =& (\DD_{w} - (\TT- \II)\DD_{w}\TT^{-1} - \AA'\DD_{w}\DD'^{-1}\TT^{-1}) \DD_{w}^{-1}\II\\
				    =& (\DD_{w} - (\TT- \II)\DD' - \AA') \DD_{w}^{-1}\II\\
				    =& (\DD_{w}-\WW)\DD_{w}^{-1}\II\;,
\end{align*}
where the interaction matrix $\WW$ is the reweighed $\AA'$ plus the self loops represented by the diagonal matrix $(\TT- \II)\DD'$, with $\DD_{w}$ represents its diagonal degree matrix. Delay transformation allows us to rescale different $\TT$ to $\II$. A simple special case is when $\TT = \alpha \II$ is a scalar matrix. It can be understood as rescaling the global time of that layer.
\end{proof}
\noindent{\em Remark}:
Giving the dynamical parameters $B$ and $\TT$, we can always find the interaction matrix $\WW$ by
$$\WW = B\AA + (\TT- \II)\DD'\;.
$$
Although by Lemma~\ref{the:mapping}, it is not unique, the interaction matrix $\WW$ serves the purpose of unifying dynamics across the layers.

\begin{theorem}
\label{theorem:transformUN}
For an undirected network $G =\\ (V,E,\AA)$, the dynamics $\LL=(\DD'-B\AA B)(\DD'\TT)^{-1}$ is equivalent to a URW on another transformed graph.
\end{theorem}

\begin{proof}
The bias transform in the undirected case is discussed in \cite{lambiotte_flow_2011, Ghosh2014KDD}. A biased random walk with transition probability $P_{ij} \propto b_i a_{ij}$ is equivalent to an unbiased random walk on a reweighed adjacency matrix:
$a'_{ij} = b_i a_{ij} b_j \;.$ Notice that the matrix products on both sises ensure that the resulting random walk is detailed balanced. The delay transformation remains the same as in Theorem~\ref{theorem:transform}.
\end{proof}

By Lemma~\ref{the:mappingUn}, the layer transformations of an undirected graphs produce a unique interaction matrix
\begin{corollary}
For an undirected network $G =\\ (V,E,\AA)$, and the dynamical parameters $B, \TT$, the interaction matrix
$$\WW = B\AA B + (\TT- \II)\DD'\;.
$$
is unique up to a global scaling factor.
\end{corollary}

With the transformed layers $\WW^1, \WW^2, ...\WW^l$ now all underly the same simple URW, we are ready to discuss the second stage of the framework in \figref{fig:framework}: how to combine them into a joint structure.

\section{Inter-layer composition}
\label{sec:inter}

{\em Inter-layer composition}
is the problem of constructing inter-layer edges that connects the transformed layers into a coherent structure. We base our composition on the observed inter-layer dynamics. Depending on the data source and problem of interest, the way layers interact with each other differs. We will consider several commonly seen situations in social and technological networks.

\subsection{Multiplex composition}
We start with the simplest case of all, when inter-layer structures are absent. In this case, simple matrix addition does the trick, and we have the following algorithm

\begin{algorithm}
\caption{Multiplex network composition}
\label{alg:multiplex}
\noindent {\textbf{Input}:} weighted network layers: $G^1= (V, E^1,\AA^1), G^2= (V, E^2, \AA^2), ..., G^l= (V, E^l, \AA^l)$, parameters of the dynamics: $\TT^1,B^1,\TT^2,B^2,..., \TT^l, B^l$,\\
{\textbf{Algorithm}}
\begin{itemize}
\item Apply the reweighing transformation $\AA'^i = B^i\AA^i$ ($\AA'^i = B^i\AA^i B^i$ for undirected graphs)
\item Apply the scaling transformation $\WW^i = \AA'^i + (\TT^i- \II)\DD'^i$
\end{itemize}
{\textbf{Output}} the $n\times n$ adjacency matrix $\WW = \sum_{i=1}^l \WW^i$
\end{algorithm}

\begin{comment}
In multiplex networks, there are only one instance of the same vertex, or equivalently, the ``inter-layer instances'' are always in the same state. To simulate a multiplex inter-layer connection in $\WWW$, one simply enforce 

To see how partial multiplexity can be modeled in $\WWW$, we again take a dynamical view. 

In multiplex networks, there are only one instance of the same vertex, or equivalently, the ``inter-layer instances'' are always in the same state. 

We can simulate a multiplex inter-layer connection by making its edge weight approaching infinity. Taking the toy example of \figref{fig:joint}, if we set the inter-layer edge weights of Alice to infinity, her phone, email and Facebook random walk probability will be synchronized at all times. Although dynamics elsewhere will be infinitely slower, but theoretical analysis show that 
\cite{gomez_diffusion_2013}
\end{comment}

\subsection{Multi-layer composition}
While we recommend \algref{alg:multiplex} for purely multiplex networks, we need a more general framework when inter-layer structures do matter. Consider the following mathematical problem:

\begin{formulation}[Super-adjacency Composition]\label{prob:basic}
%\noindent {\textbf{Input}:} weighted network layers:
Given $l$ transformed layers 
   $G^1= (V, E^1,\WW^1), ..., G^l= (V, E^l, \WW^l)$, and
egocentric inter-layer dynamics $(\MM_v : v\in V)$,
%$n$ $l\times l$ inter-layer Markovian matrices $\MM_u$ for each vertex $u\in V$.\\
compose a  $(ln\times ln)$  weighted super-adjacency matrix, 
 where $n = |V|$,
$$\mathbb{W} =  \begin{bmatrix}
                 \WW^1 &\WW^{12} &... &\WW^{1l}\\
                 \WW^{21} &\WW^{2} &... &\WW^{2l}\\
                 ...\\
                 \WW^{l1} &\WW^{l2} &... &\WW^{l}
                \end{bmatrix}
$$
to integrate the multi-layer network data.
In addition, we require all off-diagonal blocks of $\WWW$ are diagoal matrices.
In other words,  $\WWW$ represent a diagonal multi-layer networks,
  as defined in \cite{kivela_multilayer_2013}, which
  means that all inter-layer edges are
   between the same vertex at different layers.
\end{formulation}

\noindent{\em Remark}: Here, in $\mathbb{W}$,
  the $l$ diagonal ($n\times n$)-blocks are directly fed from 
  the first stage $\WW^1, \WW^2, ..., \WW^l$.
  
We have used the model of egocentric inter-layer dynamics for each vertex $(\MM_v : v\in V)$, with $\MM_v$ being the stochastic transition matrix for the inter-layer instances of the same vertex $v$. Such egocentric models are considered to be fundamental in the formation of social structures\cite{dunning_egocentric_1992,boyd_social_2007}, and might be readily available from existing social studies. They are also easy to crawl in social networks that provide cross-platform interfaces.

Together with the traditional horizontal perspective in \figref{fig:example}, egocentric inter-layer dynamics form a vertical perspective of the same joint system, where a unified dynamical process unfolds. For illustration, consider our toy example of \figref{fig:joint}.
Suppose when Alice receives a message 
  from a phone call, she might pass on the
  message directly by calling with probability $0.6 = 0.4+0.2$,
  or relay the message through emails with probability $0.3$,
  or post it on a Facebook wall with probability $0.1$.

%Specifying all the $\MM_v$ matrices might seems excessive and unrealistic, but it provides us a general mathematical framework. In the following subsections, we will discuss several special cases.

\begin{figure}
\begin{center}
  \begin{subfigure}[b]{0.45\textwidth}
    \includegraphics[width=\textwidth]{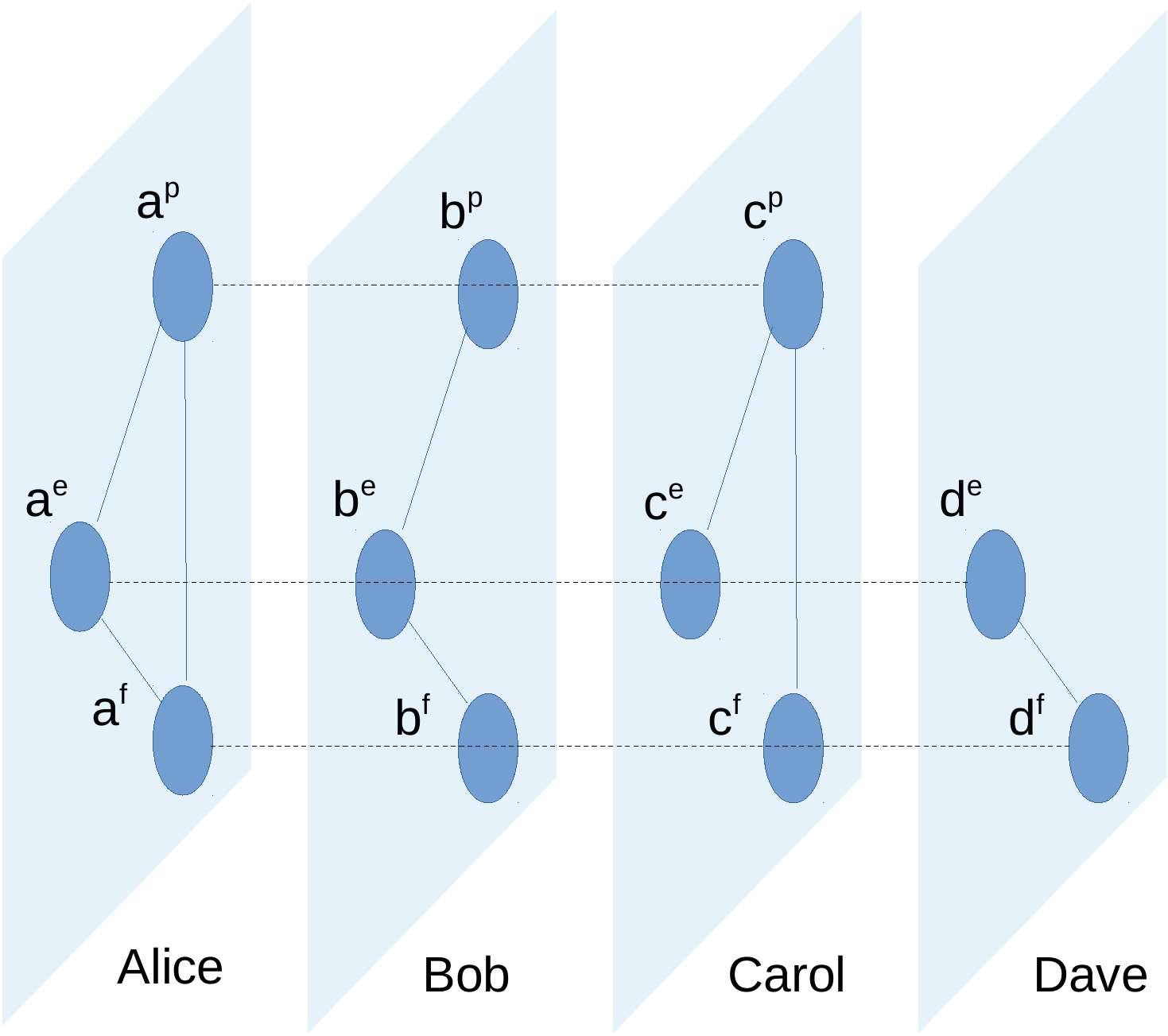}
    \caption{The toy example in vertical perspective}
  \end{subfigure}
  \begin{subfigure}[b]{0.48\textwidth}
    \includegraphics[width=\textwidth]{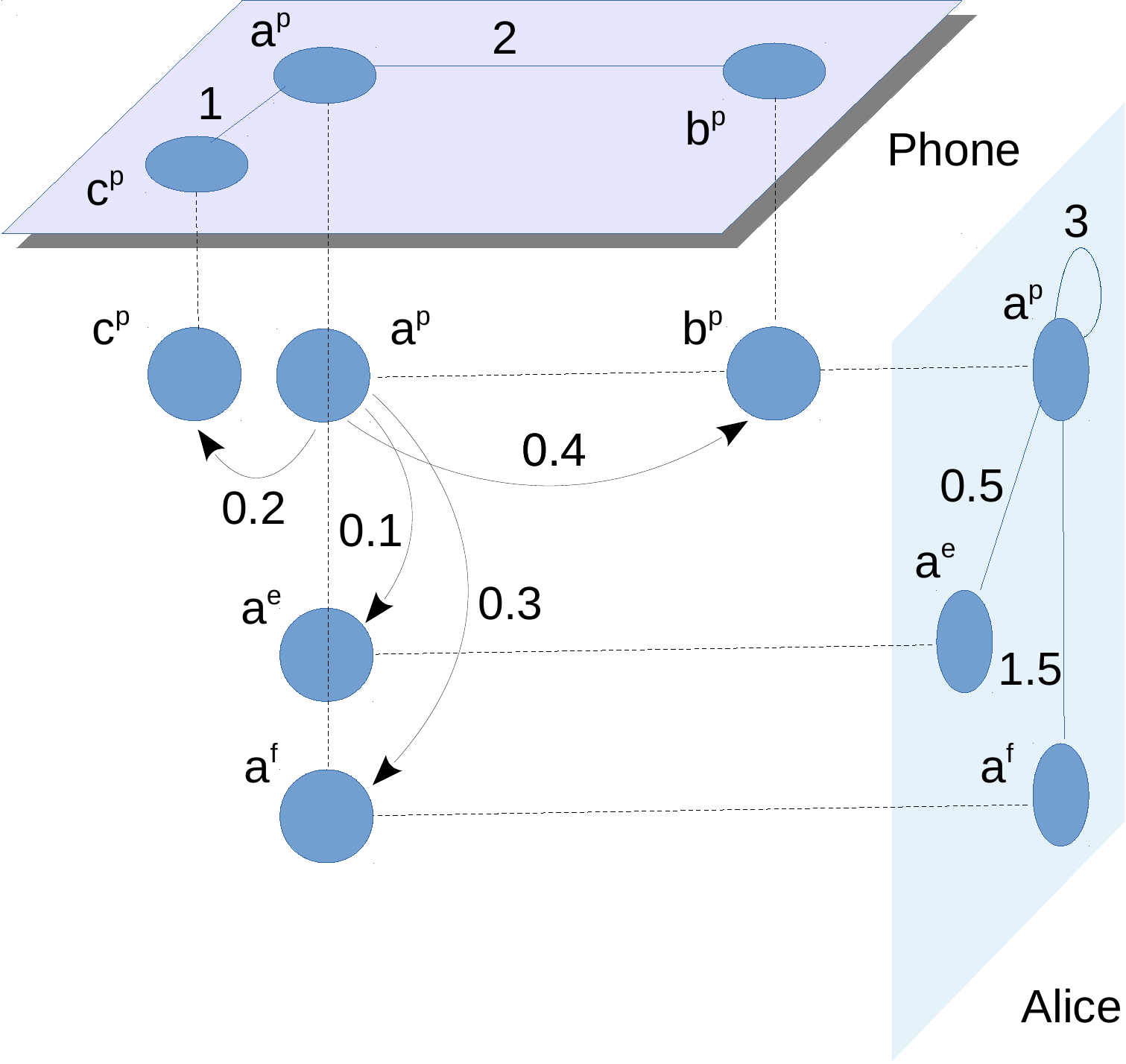}
    \caption{The Markov random walk from vertex $a^p$}
  \end{subfigure}
\end{center}
\caption{Inter-layer dynamics of the toy example}
\label{fig:joint}
\end{figure}

Our plan is to first formulate the network composition problem in the complete-information setting. The mathematical characterizations of this ideal setting can then be used to find feasible solution spaces when only partial-information is available. The additional degrees of freedom will also allow us to optimize the design space of inter-layer edges, or predict missing links. In this section, we will first show that a
  ``Dynamic view of network composition'' leads a
    feasible and unique formulation of $\mathbb{W}$.

\begin{quote}
{\em Which composed super-adjacency matrix $\mathbb{W}$
  properly integrates the multi-layer network data with 
   egocentric inter-layer dynamics?}
\end{quote}

%We will base our formulation on a dynamical view of the network. 
To answer this questions,
recall that the input is specified by
  $l$ transformed adjacencies $\WW^1,...,\WW^l$, and 
   $n$ egocentric Markov models $(\MM_v: v\in V)$.
In the ``dynamical view'', by Lemma~\ref{the:mapping}
   each layer $\WW^i$ also uniquely defines a Markov model, $\MM_{\WW^i}$.
Together, these $l+n$ Markov models define a joint Markov model, whose
  ajacency structure is the desired super-composition.
Thus, we aim to identify a weighted $(ln\times ln)$-adjacency matrix
  $\mathbb{W}$, whose random-walk Markov model, $\MM_{\WWW}$,
  satisfies the following two basic conditions:
\begin{enumerate}
\item  {\bf Layer Consistency}: The random-walk Markov model
  of each layer,  $\MM_{\AA^i}$, $i\in [1,2,..,l]$,  is the {\em projection}
   of $\MM_{\WWW}$ to that layer, and 
\item {\bf Ego Consistency}: 
     The egocentric inter-layer dynamics, $\MM_v$, of vertex $v\in V$, 
     is the {\em layer marginals} of $\MM_{\WWW}$ at vertex $v$.
%can be ``consistently individualized'' into all egocentric    
\end {enumerate}
%Recall the random-walk process on $G^i$  is given by the $(n\times n)$-transition matrix  $\PP^i = \AA^i (\DD^{i})^{-1}$,  where $\DD^{i}$ is the diagonal matrix of out-degrees:   $d^i_u = \sum_v \AA^i_{uv}$.
%Now, let $\PPP$ be the Markov random-walk model of $\WWW$.

%Considering intra and inter-layer structures as random walks leads to a mathematically principled framework for network
%  composition. 
%% When inter-layer dynamics are
%% relevant, the combined layers operates in the union of their vector
%% spaces, and the result is a $ln\times ln$ super adjacency matrix
%% $\mathbb{W}$.
%The composed ajacency matrix $\mathbb{W}$ should define
% then underlies 
%  a joint Markov random walk $\mathbb{P}$
%  that is consistent with these random-walk processes:
%We will now formally define the notion of projection and layer marginals.
Recall that the projection of a Markov model $\MM$ onto a subset 
  is simply  the stochastic normalization of corresponding 
  principal submatrix of $\MM$.
Thus, Condition 1 is automatically achieved
 by setting diagonal blocks of $\WWW$ as $\WW^1,...,\WW^l$ in Formulation \ref{prob:basic}.
%% We
%% only consider in this paper the ``diagonal multi-layer networks''
%% according to the definitions in \cite{kivela_multilayer_2013}，which
%% means we only allow inter-layer edges between the same vertex on
%% different layers. \figref{fig:joint} shows a step of $\mathbb{P}$ from
%% the vertex $a^p$, or Alice on the phone layer in the toy example
%% \figref{fig:example}.

Condition 2 addresses egocentric inter-layer dynamics.
Notice that for each $v\in V$,  $\WWW$ defines an $l\times l$ interlayer
  adjacency  matrix $\WW_v$. % as its inter-layer interaction model for node $u$.
The random-walk process, $\MM_{\WW_{v}}$,
  is the projection of the joint Markov process $\MM_{\WWW}$
  to the vertical slice consists of instances of $v$ in different layers.
Condition 2 then requires that $\MM_{\WW_{v}}$ should be consistent with
   $v$'s egocentric inter-layer dynamics $\MM_v$.
   
To be more specific, let $q_{v,i}$ denote the 
  transition probability according to $\MM_{\WWW}$
  for going from vertex $v$ in the $i^{th}$ layer to 
  some $u$ in the same layer.
%Then we say
%%  which can be better
%%    expressed by a one-step random walk.
%% Note that the Markov process $\M\MM_u$
%% Consider an arbibrary $l$-place distribution $\ttheta_u$.
%%   $\M\MM_u \ttheta_u$ specifies the one-step layer-distribution
%%   according to $u$'s inter-layer dynamics.
Let $\QQ_v$ be the $l\times l$ diagonal matrix
  of $[q_{v,i}: i\in [l]]$.
Then, $\QQ_v + \MM_{\WW_{v}} \cdot (\II - \QQ_v)$
  denote the 
  layer marginals of the joint Markov model $\MM_{\WWW}$
  at  vertex $v$.
Consequently, 
 Condition 2 requires that layer marginals
 $\MM_v = \QQ_v + \MM_{\WW_{v}} \cdot (\II - \QQ_v)$. Intuitively, egocentric inter-layer dynamics $\MM_v$ bridges between the orthogonal projections by including $\QQ_v$ as well as $\MM_{\WW_{v}}$.

Now we ready to present the main theorem of this paper:
\begin{theorem}
For any multi-layer data $(\AA_i: i\in [l], \MM_{v}: v \in V)$,
 there exists a unique and feasible super-composition $\WWW$ 
 that satisfies both Layer Consistency and Ego Consistency.
%%  of a Markovian consistent solution to Problem \ref{prob:basic} in directed graphs.
\label{th:directedFull}
\end{theorem}
\begin{proof}
Because Formulation \ref{prob:basic}
  requires that all off-diagonal blocks of $\WWW$ are diagonal matrices,
  we have $(l^2-l)n$ degrees of freedom  
  after meeting Condition 1.
  
Notice that $(\MM_v: v\in V)$ are $n$ stochastic $l\times l$ matrices.
Thus, Condition 2 represents $(l^2-l)n$ dimensional constraints, which matches perfectly with the remaining degrees of freedom. Uniqueness proven.

To prove the feasibility of the unique solution, we introduce the algorithmic framework \algref{alg:superAdj}, 
\begin{algorithm}[H]
\caption{Multilayer network composition}
\label{alg:superAdj}
\noindent {\textbf{Input}:} weighted network layers: $G^1= (V, E^1,\AA^1), G^2= (V, E^2, \AA^2), ..., G^l= (V, E^l, \AA^l)$, parameters of the dynamics: $\TT^1,B^1,\TT^2,B^2,..., \TT^l, B^l$,  and $n$ $l\times l$ egocentric inter-layer Markovian matrix $M_u$ for each vertex $u\in V$.\\
{\textbf{Algorithm}}
\begin{itemize}
\item Apply the reweighing transformation $\AA'^i = B^i\AA^i$ ($\AA'^i = B^i\AA^i B^i$ for undirected graphs)
\item Apply the scaling transformation $\WW^i = \AA'^i + (\TT^i- \II)\DD'^i$
\item Create a $ln\times ln$ empty matrix $\mathbb{W}$
\item Fill the $l$ diagonal blocks (each of size $n\times n$) with $\WW^1, \WW^2, ..., \WW^l$
\item Construct the off diagonal blocks $\WW^{ij}$ (each of size $n\times n$) for all layer pairs $i$ and $j$ based on \algref{alg:interEdges} with $\WW^1, \WW^2, ..., \WW^l$ as inputs
\end{itemize}
{\textbf{Output}} The super adjacency matrix
$$\mathbb{W} =  \begin{bmatrix}
                 \WW^1 &\WW^{12} &... &\WW^{1l}\\
                 \WW^{21} &\WW^{2} &... &\WW^{2l}\\
                 ...\\
                 \WW^{l1} &\WW^{l2} &... &\WW^{l}
                \end{bmatrix}
$$
\end{algorithm}

We need a subroutine \algref{alg:interEdges} to satisfy inter-layer constraints at each node.
  we rearrange the row and column of $\mathbb{W}$
   so that the counterparts of the same 
  vertex are grouped together.
The rearrangement express $\WWW$ with the following block structures:
\[\bar{\mathbb{W}} = \begin{bmatrix}
    \WW_1 	&\WW_{12}	&...	&\WW_{1n}\\
    \WW_{21} 	&\WW_2 		&... 	&\WW_{2n}\\
    ...\\
    \WW_{n1} 	&\WW_{n2}	&... 	&\WW_n,
  \end{bmatrix}
\]
where $\WW_{u,v}$ are $l\times l$ matrices that have already been fixed
  by Condition 1.
The $n$ matrices, $\WW_v: v\in V$ on the diagonal blocks,
  contains all entries that we will need to set using Condition 2.
Because the rearrangement of $\WWW$ preserves the diagonal
  entries up to reordering, the diagonal entries of $\WW_v$, $v\in V$,
  are also set by Condition 1.
The rest $n(l^2-l)$ entries lead to the same degrees of freedom we discussed earlier.

%To complete the proof, we will show that for all $l\times l$ stochastic matrices $\MM_v$,

%algorithm below shows that for all there exists unique pair of diagonal matrix $\QQ_v$ and stochastic matrix $\bar{\MM}_v$, such that \[\MM_v = \QQ_v + \bar{\MM}_{v} \cdot (\II - \QQ_v)\]

%% By the ``diagonal'' inter-layer structure assumption, we know that both the original $\WW^{ij}$ ($n\times n$) and reordered $\AA_{uv}$ ($l\times l$) off diagonal blocks are diagonal matrices. The $\AA_{uv}$ blocks are composed of the entries from those in the diagonal blocks $\AA^i$, while the diagonal $\WW_{u}$ blocks are composed of the entries from $\WW^{ij}$\footnote{The diagonal of $\mathbb{W}$ (self-loops) stays on the diagonal.}.

The reordered $\WW_{u}$ blocks are closely related to the egocentric
adjacencies $\XX_u$ underlying the egocentric inter-layer dynamics
$\MM_u$. The vertical slice in \figref{fig:joint} demonstrates such a
$\XX_a$, where intra-layer transitions are captured using
self-loops. Subroutine \algref{alg:interEdges} can now be specified as 

\begin{algorithm}[H]
\caption{Building inter-layer blocks}
\label{alg:interEdges}
\noindent {\textbf{Input}:} transformed layers: $G_1= (V, E^1, \WW^1), G_2= (V, E^2, \WW^2), ..., G_l= (V, E^l, \WW^l)$,  and a $l\times l$ egocentric inter-layer transition matrix $\MM_u$ for vertex $u\in V$.\\
{\textbf{Algorithm}}
\begin{itemize}
\item Create a $l\times l$ empty matrix $\XX_u$
\item Fill the diagonal elements with $\XX^{ii}_u = d^{i}_u(out)$
\item Construct the off diagonal elements
$$\XX^{ij}_u = \frac{\MM_u^{ij}}{\MM_u^{ii}}d^{i}_u(out)$$
\end{itemize}
{\textbf{Output}} Block $\XX_u$ and repeat for each $u\in V$\\
\end{algorithm}

\begin{comment}
Here we first fill the diagonals $\XX^{ii}_{u} = d^{i}_u(out)$, which is the sum of total out edge weights of vertex $u$ in layer $i$. In our example \figref{fig:joint}, the degree of Alice in the phone layer is $d^{p}_a(out) = 3$ (treat undirected edges as bidirectional). Assuming an asymmetric $\XX_u$, we still have $l^2-l$ degrees of freedom.

The off diagonal elements are calculated based on the input stochastic matrix $\MM_u$. With an unbiased random walk with uniform delays, we know that
$$\MM_u^{ij} = \frac{\XX^{ij}_{u}}{\sum_r\XX^{ir}_{u}},$$
where the denominator is the total out degree of vertex $u^i$ across the layers. Since the columns of $\MM_u$ all sum up to 1, the last row is dependent on all previous rows, giving rise to $l(l-1)$ linearly independent constrains. Considering the $l^2-l$ degrees of freedom, we have a fully determined system. On top of that, \algref{alg:interEdges} will always lead to feasible solutions with the constrains $\XX_u^{ij}\geq 0$, provided that $\MM_u$ entries are well defined. 
\end{comment}

Using Lemma \ref{the:mapping}, we can rewrite the steps in \algref{alg:interEdges} as $\XX_u = \MM_u \GGamma$, by setting the $i^{th}$ entry of $\GGamma$ uniquely as $d^{i}_u(out)/\MM_u^{ii}$. Intuitively, we are simply respecting the layer inputs and using the intra-layer dynamics to determine the vertex scaling factor.

From \figref{fig:joint}, it is clear that the off-diagonal parts of $\XX_u$ is exactly what we are looking for in $\WW_u$ blocks. Or $\WW_u = \XX_u - D_u(out)$, where the diagonal matrix $D_u(out)$ is composed of $d^{i}_u(out)$ entries. With the uniquely solvable $\XX_u$ blocks, we can now complete the output $\mathbb{W}$ by filling its off diagonal blocks $\WW^{ij}$ with reordered $\WW_u$ blocks. On top of that, \algref{alg:interEdges} will always lead to feasible solutions with the constrains $\XX_u^{ij}\geq 0$, provided that $\MM_u$ entries are well defined.
Uniqueness and feasibility proven.
\end{proof}
\subsection{Overdetermined Composition}
\label{sec:interEdges}
Based on Theorem \ref{th:directedFull}, we have a fully determined system for consistent network composition when we have complete information about the personalized inter-layer dynamics. In practice, depending on the inputs and constrains, we might have underdetermined, overdetermined or even mixed systems.

If the network is undirected, the Markov process is under the detailed balance condition. They also become reversible, leading to a additional dependency for each independent loop in $\XX_u$ because of Kolmogorov's criterion. To count the number of independent loops, we simply subtract the number of edges in a connected tree ($l-1$) from the total number of edges, as each additional edge on the tree will introduce an independent loop. In our case, we consider all possible inter-layer connections (a complete graph $\XX_u$). Then the undirected Markov matrix would lead to
$$l(l-1) - [\frac{l(l-1)}{2} - (l-1)] = \frac{l(l+1)}{2} -1
$$
constrains. The degree of freedom of an undirected marginalized adjacency $X_u$ is $\frac{l(l-1)}{2}$. There are $l-1$ more constrains than variables. Leading to:
\begin{conjecture}In undirected graphs, the existence of a Layer Consistent and Ego Consistent solution to Formulation \ref{prob:basic} depends on the inputs.\\
\label{th:undirectedFull}
\end{conjecture}

In our example \figref{fig:joint}, we have $d^{p}_a = 3$ and
$$\XX_a^{pe} = \frac{M_a^{pe}}{M_a^{pp}}d^{p}_a = \frac{0.1\times 3}{0.2+0.4} = \frac{1}{2}.$$
If $\XX_u$ is undirected as shown in \figref{fig:joint}, we have
$$M_a^{ep} = \frac{\XX_a^{ep}}{\sum_r \XX_a^{er}} = \frac{0.5}{\sum_r \XX_a^{er}}.$$
If the inputs does not satisfy the above constrain, there will not be any feasible solution to Formulation \ref{prob:basic}.

Other overdetermined systems can arise when we have some direct measures of the inter-layer structures, and approximate solutions can be found by minimizing some error terms. In real applications, however, it is much more likely that we have less empirical measures, and we will be facing systems with additional degrees of freedoms. Such underdetermined systems leave spaces for other considerations, and are often associated with network design and other optimization Formulations.

\subsection{Underdetermined Compositions}
In our combined social networks example \figref{fig:example}, we might not know each user's message routing strategies, but we can can track the marginal distribution among the layers of how the messages are propagated for each user. In mixed membership community models \cite{Airoldi_mixed_2008}, we can measure or infer the percentage of edge type each vertex is associated with. Such information is captured by the stationary distributions of $\XX_u$ (for example, the vertical slice in \figref{fig:joint}). In these situations, we can restate the Formulation as

\begin{formulation}
\label{prob:partial}
{\textbf{Network composition with inter-layer stationary  distributions}}\\
\noindent Given transformed network layers: $G^1= (V, E^1,\WW^1), G^2= (V, E^2, \WW^2), ..., G^l= (V, E^l, \WW^l)$, and the stationary distributions over the layers $\pi_u$ for each vertex $u\in V$,
compose a super adjacency matrix $\mathbb{W}$ satisfying Markovian consistency.
\end{formulation}

\algref{alg:superAdj} remains a valid framework for Formulation \ref{prob:partial}. However, there will generally be no unique solution with \algref{alg:interEdges} as the subroutine. The inter-layer stationary  distributions amounts to $l-1$ constrains by considering the normalization conditions. This leads to
$$l(l-1) - (l-1) = (l-1)^2
$$
degrees of freedom for each user. A algorithmic solution to this underdetermined system can be specified as \algref{alg:stationary}.

\begin{algorithm}
\caption{Building inter-layer block $\XX_u$ with stationary distributions}
\label{alg:stationary}
\noindent {\textbf{Input}:} weighted network layers: $G_1= (V, E^1, \AA^1), G_2= (V, E^2, \AA^2),..., G_2= (V, E^l, \AA^l)$,  and n $l\times 1$ stationary distribution vector $\pi_u$ for each vertex $u\in V$.\\
{\textbf{Algorithm}}
\begin{itemize}
\item Create a $l\times l$ matrix $\XX_u$ with $l^2$ free variables
\item Constrains the diagonal elements with $\XX^{ii}_u = d^{i}_u(out)$
\item Solve for the off diagonal elements with the constrains
$$\forall i \neq j, \frac{\pi_u^i}{\pi_u^j} = \sum_{\forall r}\frac{\XX^{ir}_u}{\XX^{jr}_u}$$
\end{itemize}
{\textbf{Output}} Block $\XX_u$ and repeat for each $u\in V$\\
\end{algorithm}

If we assume the underlying network is undirected, the symmetry will reduce the degree of freedom to
$$\frac{l(l-1)}{2} - (l-1) = \frac{(l-2)(l-1)}{2}.
$$
With $l=2$, we recover a fully determined system. In \algref{alg:stationary}, the solution will be
$$\XX^{12}_u = \frac{\pi_u^1 d^{2}_u - (1-\pi_u^1)d^{1}_u}{(1- 2\pi_u^1)} = \frac{\pi_u^1 (d^{2}_u+d^{1}_u) - d^{1}_u}{(1- 2\pi_u^1)}.
$$
For it to be feasible, we need $\XX^{12}_u\geq 0$ or
$$0.5 \leq \pi_u^1 \leq \frac{d^{1}_u}{d^{1}_u+d^{2}_u}\ or\  0.5 \geq \pi_u^1 \geq \frac{d^{1}_u}{d^{1}_u+d^{2}_u}.$$

For underdetermined systems in general, we can specify an optimization objective function and use the additional degrees of freedom for network designs. Assuming the feasible solutions to \algref{alg:stationary} form a family of $\{\WWW\}$ whose random-walk Markov process is consistent with the constrains, we may have
\begin{itemize}
 \item Minimum network volume: $\min_{\WWW\in\{\WWW\}} \sum_{u,i,j} \WW_u^{ij} $;
 \item Maximum conductance: $\max_{\WWW\in\{\WWW\}} \min_{S\in \mathbb{V}}\frac{cut(S)}{\sum_{u\in S}d_u^{out}}$, where $\mathbb{V}$ is the super composed vertex set with layer copies;
\end{itemize}
as potential objective functions.

Another common scenario leading to underdetermined systems is when we have inter-layer distance measures. Fully specified pair-wise layer distances amounts to $l(l-1)-1$ constrains, leaving a single degree of freedom which can also be interpreted as the scaling of the inter-layer edge weights relative to their intra-layer counterparts. 

This construction is particularly suitable for combining time series of networks into multilayer structures. The temporal structure forms a one dimensional line. In the case when all vertices in the same layer shares the same time stamp, and we measure the distance between layers simply by the time difference, inter-layer adjacencies $\XX_u$ become the same for all vertices. Similarly, the degree of freedom for each vertex combines into a global parameter of inter-layer strengths. Such ``layer coupled'' multilayer structures have appeared in many previous studies as we discussed in \secref{sec:introduction}.

We will demonstrate how \algref{alg:stationary} and inter-layer distance models might be applied to real data set in \secref{sec:exp}.

\section{Empirical examples}
\label{sec:exp}
We apply the framework to study real world data sets. We demonstrate that community structure in a multi-layer network is sensitive to details of the inter-layer and intra-layer dynamics.
%We will demonstrate the importance of consistency in multi-layer composition by checking vertex centrality and community structures. 
Community structure is produced through graph bisections using the sweeping algorithm in \cite{Ghosh2014KDD}.
\subsection{Impact of layer transformations}
\begin{figure}
    \includegraphics[width=0.42\textwidth]{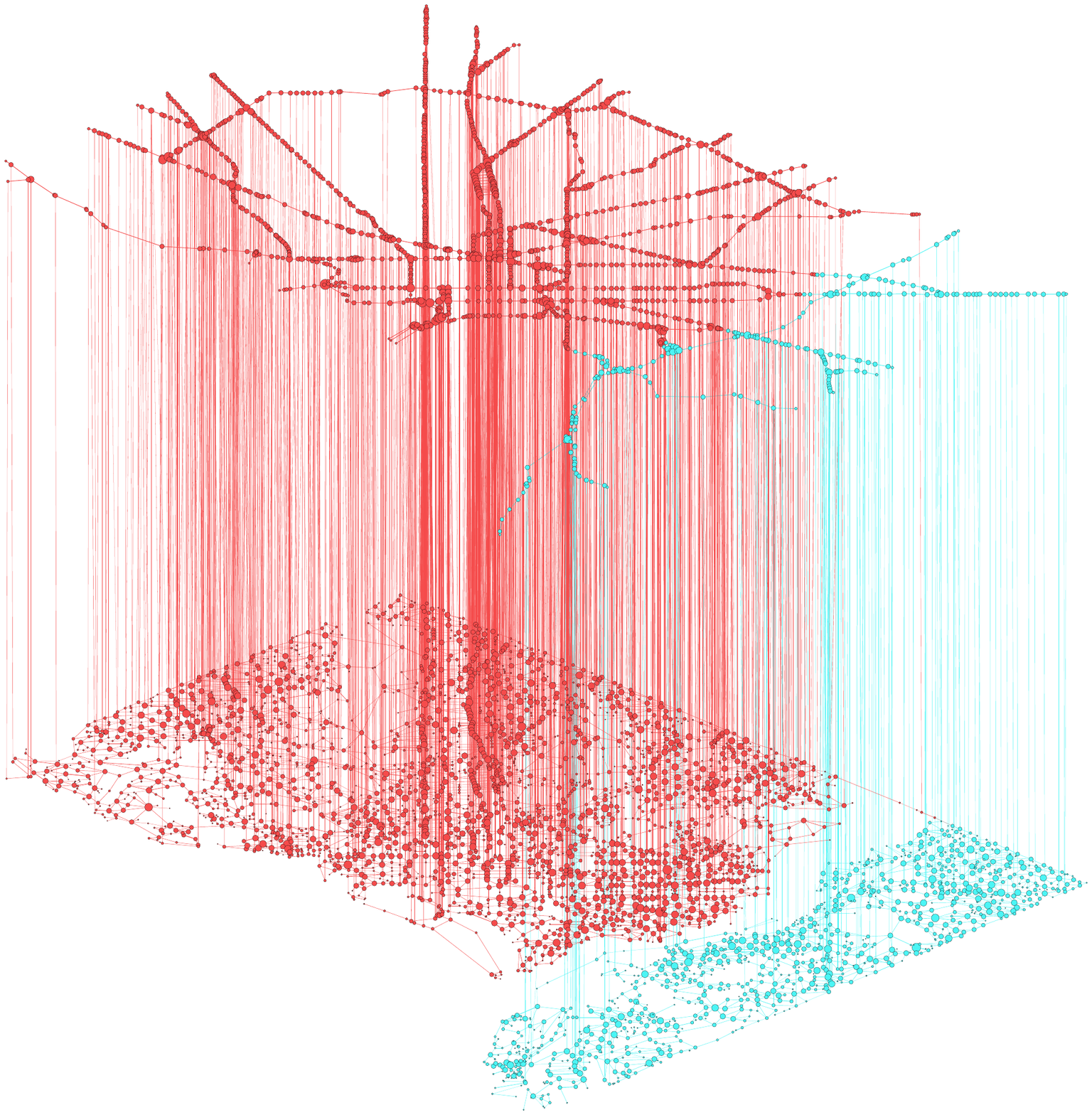}
    \includegraphics[width=0.42\textwidth]{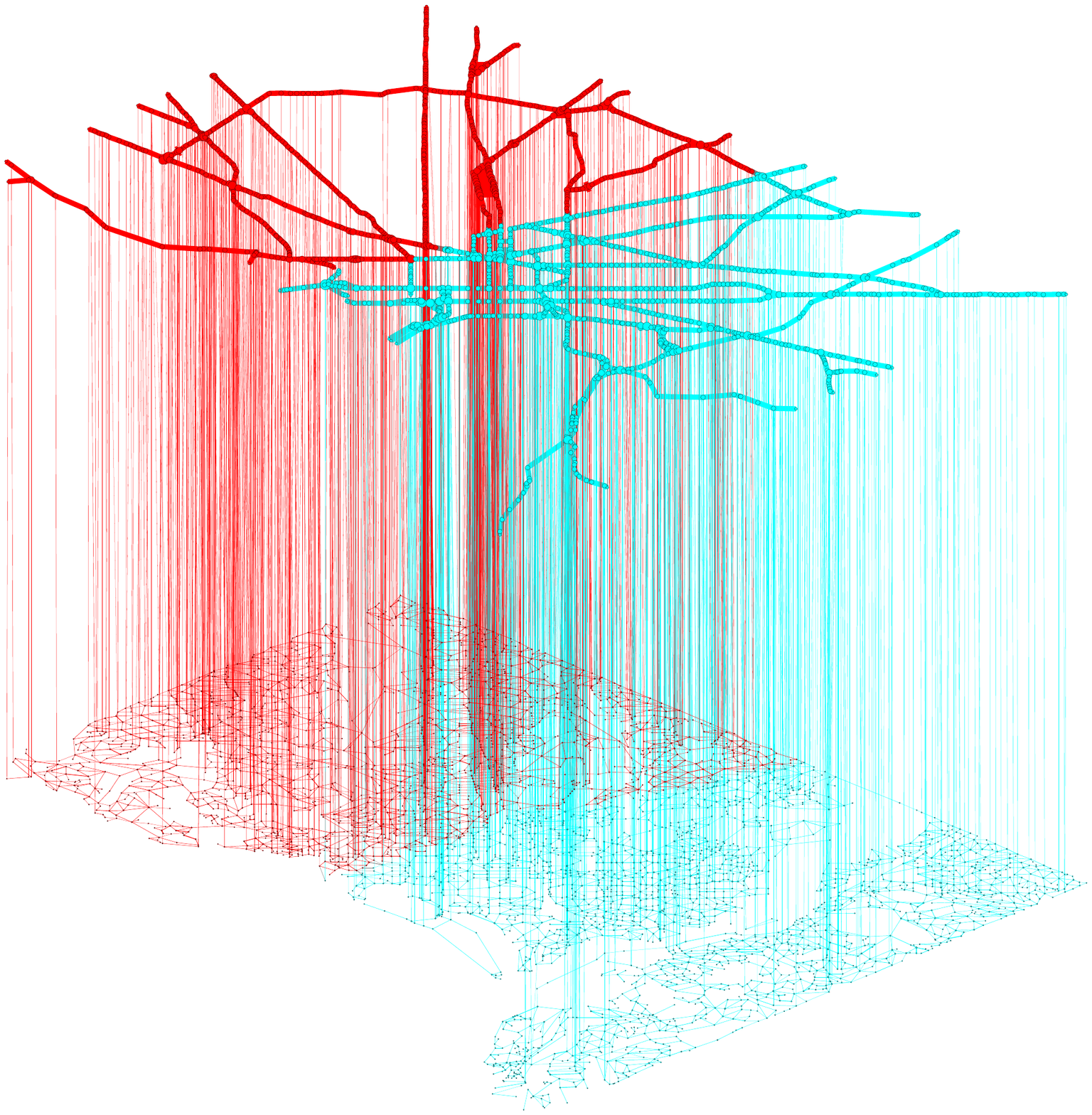}
    \caption{Bisection of road networks in DC with different composition rules}
\label{fig:road}
\end{figure}

To illustrate how layer transformation affect the structure of the resulting multi-layer network, we use the road network in the city of Washington DC. As shown in \figref{fig:road}, we modeled local roads and highways as two layers, with inter-layer edges representing highway entrances and exits. The data are based on the attributed undirected network provided by the 9th DIMACS implementation challenge\textendash Shortest Paths \cite{demetrescu20069th}. The top highway layer combines road category $A1, A2, A3$. Inter-layer edges are constructed by matching vertex labels at both layers. After removing the disconnected components, we have a multi-layer network with 10834 vertices and 28137 edges.

For comparison, we first used a conventional construction. Each highway connection received a weight of 2.0 while a local road or a highway entrances and exits received weights of 1.0. This is based on the average speed estimates listed by road categories \cite{demetrescu20069th}. Applying the graph bisection algorithm, we identified the traffic bottleneck along the Anacostia River as demonstrated on the top of \figref{fig:road}. 

In contrast, when we applied our general composition framework \algref{alg:superAdj} to the dataset, a very different picture emerged. Using the total degree as a measure, the conventional construction leads to a $12.7\%$ traffic load on the highways. On the bottom of \figref{fig:road}, we scaled the high way layer by a constant factor of 3.14, leading to a $20\%$ highway traffic load. We also introduced traffic delays at each intersection, with a delay factor $\tau$ proportional to its degree. We lack empirical observations for inter-layer Markovian matrices $M_u$. Here we made a simple assumption that vertices with both highway and local road accesses all follow a $20-80\%$ inter-layer stationary distribution with 4-times the traffic on the highway layer. This allows us to use \algref{alg:stationary} in the last step of \algref{alg:superAdj} and recover a fully determined system.

Using the same graph bisection algorithm, the new composition finds the traffic bottleneck at the center of the city, as demonstrated on the bottom of \figref{fig:road}. Based on more realistic traffic patterns, our composition framework puts more weight on the highway layer whose structure is less bottlenecked by the Anacostia River.

\subsection{Impact of inter-layer edges}
We illustrate the impact of inter-layer compositions using a multilayer coauthorship networks. Specifically, we represent coauthorship networks over time as a multi-layer network,  where each layer corresponds to a snapshot of the coauthorship network at some time.
As in other real world applications, however, the inter-layer dynamics is difficult to specify. Here we use the interlayer distance approach introduced in \secref{sec:interEdges}.

\figref{fig:coauthor} presents a collaboration networks centered around four authors: Shang-hua Teng, Daniel Spielman, Gary Miller and Kristina Lerman, as well as their coauthors on papers appearing in the ACM Digital Library. Each layer represents a separate time period: from bottom to top, it is 1985-1994, 1995-2004, and 2005-2014. The weight of an intra-layer edge represents the number of times two authors collaborated during that time period, while inter-layer edges connect the same author between neighboring decades, with weights reflecting the relative time distances.

\begin{figure}
  \begin{subfigure}[b]{0.32\textwidth}
    \centering
    \includegraphics[width=\textwidth]{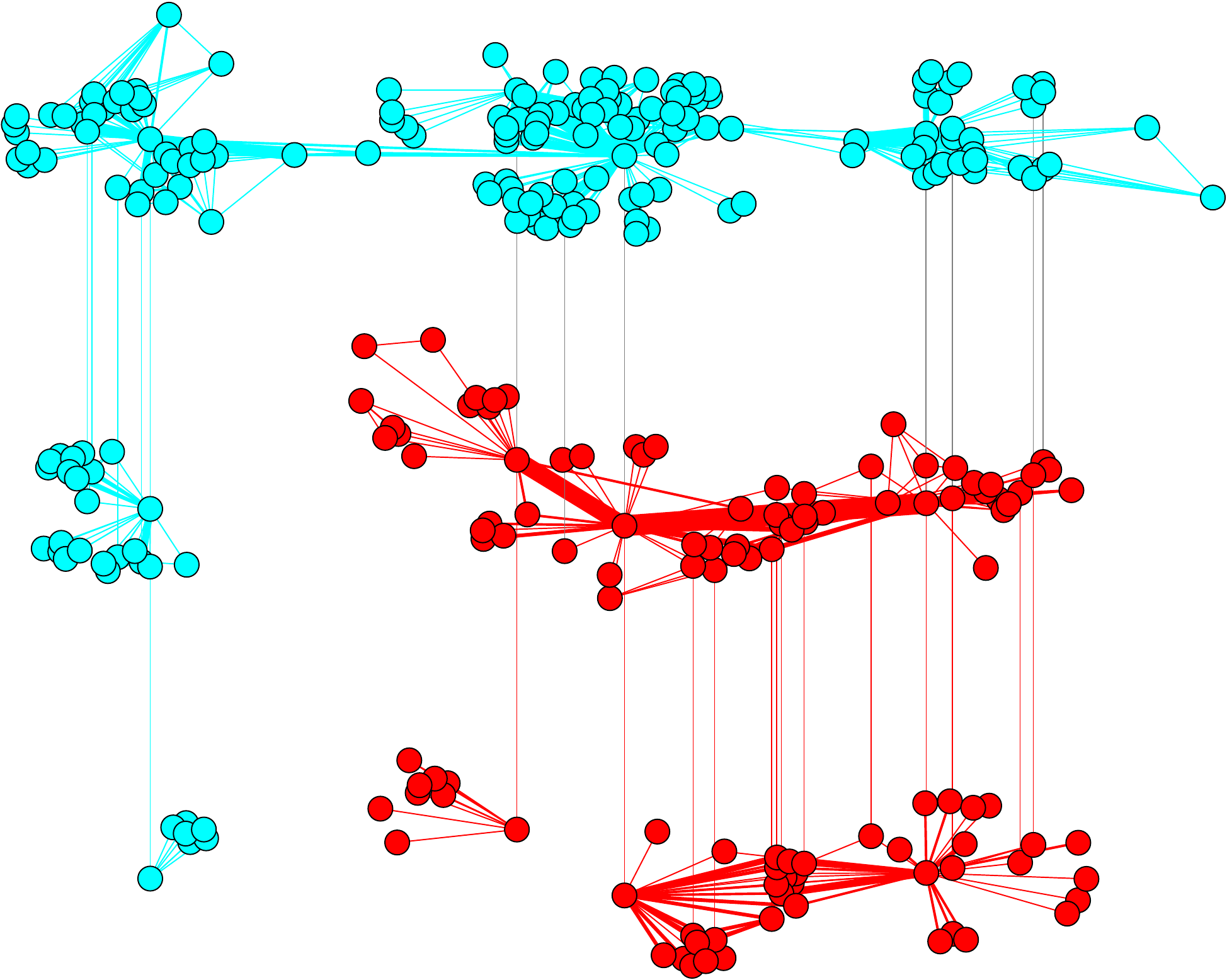}
    \caption{}
  \end{subfigure}
  \begin{subfigure}[b]{0.32\textwidth}
    \centering
    \includegraphics[width=\textwidth]{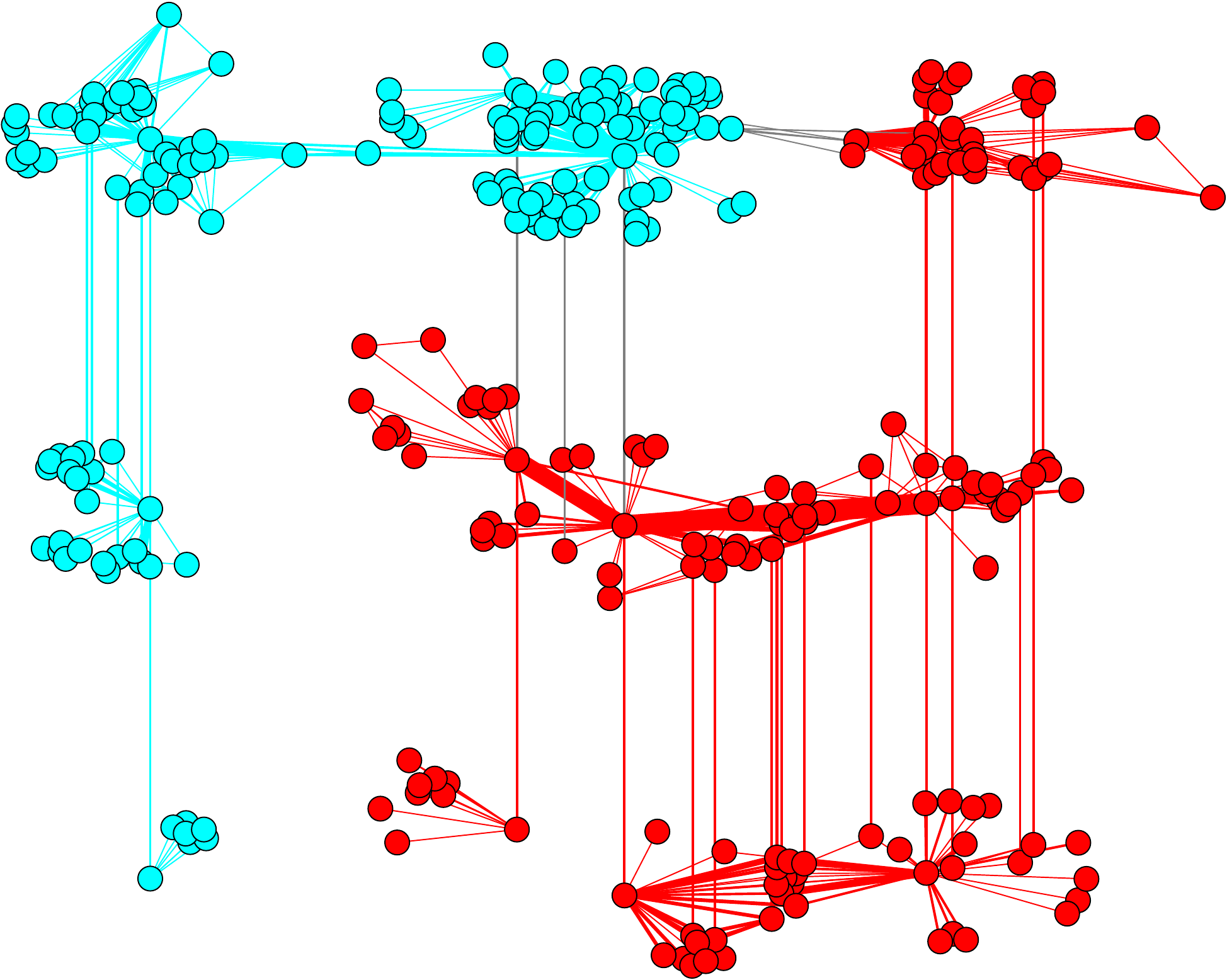}
    \caption{}
  \end{subfigure}
  \begin{subfigure}[b]{0.32\textwidth}
    \centering
    \includegraphics[width=\textwidth]{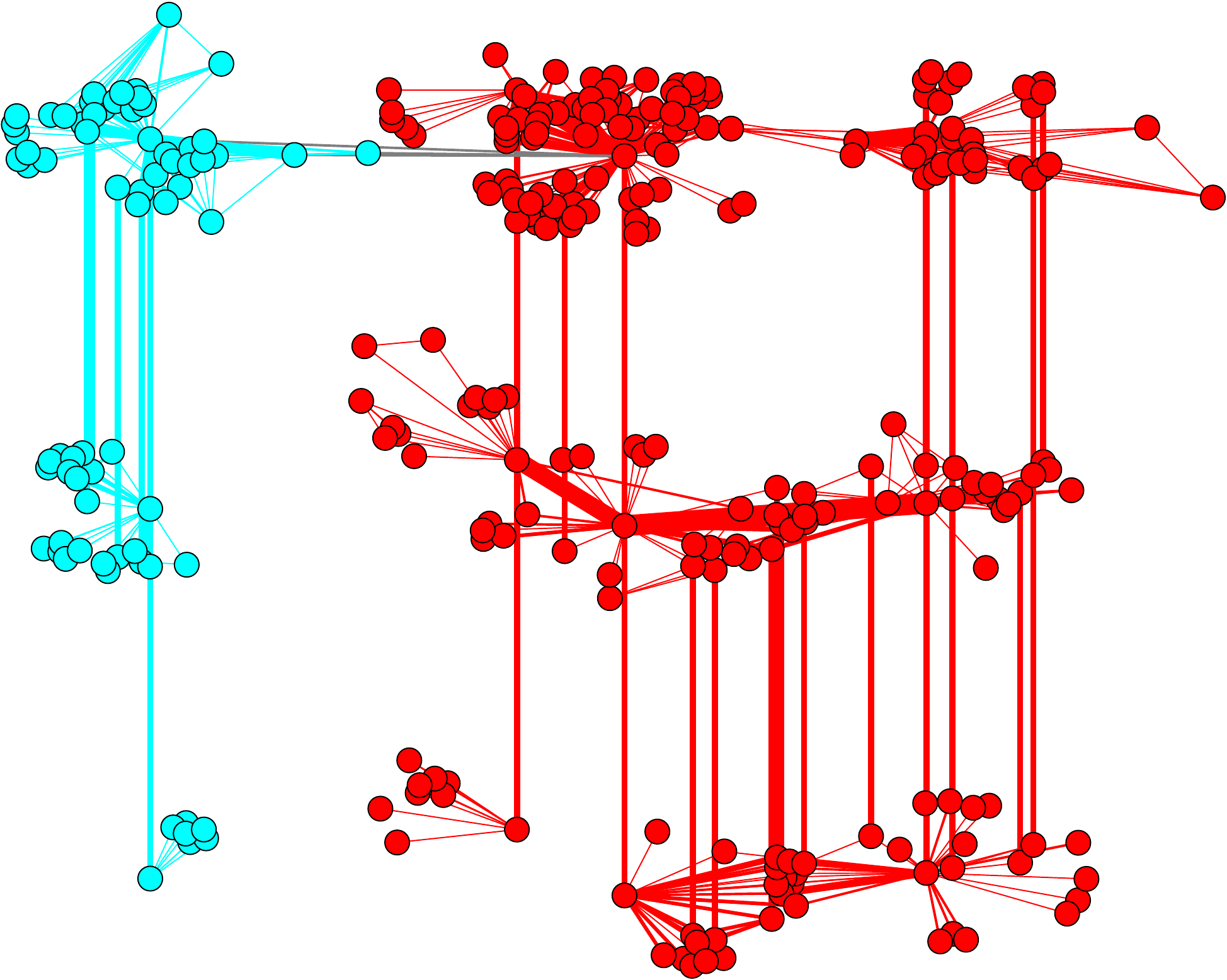}
    \caption{}
  \end{subfigure}
\caption{Bisection of coauthor networks using different inter-layer strengths}
\label{fig:coauthor}
\end{figure}

\begin{comment}
\begin{figure}
  \begin{subfigure}[b]{0.48\textwidth}
    \centering
    \includegraphics[width=\textwidth]{figures/coauthorLnew-crop}
    \caption{}
  \end{subfigure}
  \begin{subfigure}[b]{0.48\textwidth}
    \centering
    \includegraphics[width=\textwidth]{figures/coauthorRnew-crop}
    \caption{}
  \end{subfigure}
\caption{Bisection of collaboration networks using different inter-layer strengths}
\label{fig:coauthornew}
\end{figure}
\end{comment}

Since the distance between layers is uniform, all inter-layer edges have the same weight. Changing the global parameter of inter-layer edge strength, %models how likely an author is to stay in a research topic throughout the years, 
we get three different bisections of the network into communities using the normalized Laplacian \cite{Ghosh2014KDD}. In \figref{fig:coauthor}(a), the weight of each inter-layer edge 0.5 (half of the weight of a collaboration), resulting in a mostly horizontal bisection. In \figref{fig:coauthor}(c), the weight of inter-layer edges is 5.0, resulting in a vertical bisection, as separating layers is much more expensive. The most interesting case is when we set the weight on inter-layer edges to 1.0. For the earlier two decades, authors surrounding Shang-hua Teng and his Ph.D. advisor Gary Miller forms the red community largely consist of theoretical computer scientists. In the latest decade, however, the algorithm put Shang-hua Teng into the cyan group together with
Kristina Lerman where the focus has switched to graph mining and modeling. The algorithm has consistently put Shang-hua Teng and Daniel Spielman in the same community as they collaborated extensively throughout the years.

\section{Conclusion}
In this work, we proposed a mathematically principled framework for multilayer network composition based on a unifed dynamical process. We developed theorems and algorithms to construct a joint structure that can reflect the different intra and inter-layer dynamics in the inputs. 

We also discussed and demonstrated a few practical situations when the system is not fully determined. In future works, we plan to explore approximate solutions for overdetermined systems and investigate in greater details of the associated network design and optimization problems for overdetermined systems.

\bibliographystyle{abbrv}
\bibliography{references}

\begin{thebibliography}{10}

\bibitem{acar2009unsupervised}
E.~Acar and B.~Yener.
\newblock {Unsupervised multiway data analysis: A literature survey}.
\newblock {\em Knowledge and Data Engineering, IEEE Transactions on},
  21(1):6--20, 2009.

\bibitem{Airoldi_mixed_2008}
E.~M. Airoldi, D.~M. Blei, S.~E. Fienberg, and E.~P. Xing.
\newblock {Mixed Membership Stochastic Blockmodels}.
\newblock {\em J. Mach. Learn. Res.}, 9:1981--2014, June 2008.

\bibitem{aldous2002reversible}
D.~Aldous and J.~Fill.
\newblock {Reversible Markov chains and random walks on graphs}, 2002.

\bibitem{balcan_multiscale_2009}
D.~Balcan, V.~Colizza, B.~Gon\c{c}alves, H.~Hu, J.~J. Ramasco, and
  A.~Vespignani.
\newblock {Multiscale mobility networks and the spatial spreading of infectious
  diseases}.
\newblock {\em Proceedings of the National Academy of Sciences},
  106(51):21484--21489, Dec. 2009.

\bibitem{barabasi_origin_2005}
A.-L. Barabasi.
\newblock {The origin of bursts and heavy tails in human dynamics}.
\newblock {\em Nature}, 435(7039):207--211, May 2005.

\bibitem{bazzi_community_2015}
M.~Bazzi, M.~A. Porter, S.~Williams, M.~McDonald, D.~J. Fenn, and S.~D.
  Howison.
\newblock {Community detection in temporal multilayer networks, and its
  application to correlation networks}.
\newblock {\em ArXiv e-prints}, Dec. 2015.

\bibitem{bonacich1987power}
P.~Bonacich.
\newblock {Power and centrality: A family of measures}.
\newblock {\em American journal of sociology}, pages 1170--1182, 1987.

\bibitem{Borgatti05}
S.~Borgatti.
\newblock {Centrality and network flow}.
\newblock {\em Social Networks}, 27(1):55--71, Jan. 2005.

\bibitem{boyd_social_2007}
d.~m. boyd and N.~B. Ellison.
\newblock {Social {Network} {Sites}: {Definition}, {History}, and
  {Scholarship}}.
\newblock {\em Journal of Computer-Mediated Communication}, 13(1):210--230,
  Oct. 2007.

\bibitem{buldyrev_catastrophic_2010}
S.~V. Buldyrev, R.~Parshani, G.~Paul, H.~E. Stanley, and S.~Havlin.
\newblock {Catastrophic cascade of failures in interdependent networks}.
\newblock {\em Nature}, 464(7291):1025--1028, Apr. 2010.

\bibitem{Colizza2006}
V.~Colizza, A.~Barrat, M.~Barth{\'e}lemy, and A.~Vespignani.
\newblock {The role of the airline transportation network in the prediction and
  predictability of global epidemics}.
\newblock {\em Proceedings of the National Academy of Sciences of the United
  States of America}, 103(7):2015--2020, 2006.

\bibitem{de_domenico_random_2013}
M.~{De Domenico}, A.~Sole, S.~Gomez, and A.~Arenas.
\newblock {Random {Walks} on {Multiplex} {Networks}}.
\newblock {\em ArXiv e-prints}, June 2013.

\bibitem{demetrescu20069th}
C.~Demetrescu, A.~Goldberg, and D.~Johnson.
\newblock {9th DIMACS implementation challenge--Shortest Paths}.
\newblock {\em American Mathematical Society}, 2006.

\bibitem{dunning_egocentric_1992}
D.~Dunning and G.~L. Cohen.
\newblock {Egocentric definitions of traits and abilities in social judgment.}
\newblock {\em Journal of Personality and Social Psychology}, 63(3):341--355,
  1992.

\bibitem{Fortunato10}
S.~Fortunato.
\newblock {Community detection in graphs}.
\newblock {\em Physics Reports}, 486:75--174, Jan. 2010.

\bibitem{gallotti_multilayer_2015}
R.~Gallotti and M.~Barthelemy.
\newblock {The multilayer temporal network of public transport in {Great}
  {Britain}}.
\newblock {\em Scientific Data}, 2:140056, Jan. 2015.

\bibitem{gallotti_information_2015}
R.~Gallotti, M.~A. Porter, and M.~Barthelemy.
\newblock {Information measures and cognitive limits in multilayer navigation}.
\newblock {\em ArXiv e-prints}, June 2015.

\bibitem{ghosh_rethinking_2012}
R.~Ghosh and K.~Lerman.
\newblock {Rethinking Centrality: The Role of Dynamical Processes in Social
  Network Analysis}.
\newblock {\em {CoRR}}, abs/1209.4616, 2012.

\bibitem{Ghosh2014KDD}
R.~Ghosh, S.-h. Teng, K.~Lerman, and X.~Yan.
\newblock {The Interplay Between Dynamics and Networks: Centrality,
  Communities, and Cheeger Inequality}.
\newblock In {\em {Proceedings of the 20th ACM SIGKDD International Conference
  on Knowledge Discovery and Data Mining}}, {KDD '14}, pages 1406--1415, New
  York, NY, USA, 2014. ACM.

\bibitem{gomez_diffusion_2013}
S.~G{\'o}mez, A.~D{\textbackslash}'ıaz-Guilera, J.~G{\'o}mez-Garde{\~n}es,
  C.~J. P{\'e}rez-Vicente, Y.~Moreno, and A.~Arenas.
\newblock {Diffusion Dynamics on Multiplex Networks}.
\newblock {\em Phys. Rev. Lett.}, 110(2):028701, Jan. 2013.

\bibitem{Hodas12limited}
N.~O. Hodas.
\newblock {How limited visibility and divided attention constrain social
  contagion}.
\newblock In {\em {In SocialCom}}, 2012.

\bibitem{hu_multislice_2012}
H.~Hu, Y.~van Gennip, B.~Hunter, M.~A. Porter, and A.~L. Bertozzi.
\newblock {Multislice {Modularity} {Optimization} in {Community} {Detection}
  and {Image} {Segmentation}}.
\newblock {\em ArXiv e-prints}, Nov. 2012.

\bibitem{Kempe03}
D.~Kempe, J.~Kleinberg, and E.~Tardos.
\newblock {Maximizing the spread of influence through a social network}.
\newblock In {\em {KDD '03}}, pages 137--146. ACM, 2003.

\bibitem{kivela_multilayer_2013}
M.~Kivel{\"a}, A.~Arenas, M.~Barthelemy, J.~P. Gleeson, Y.~Moreno, and M.~A.
  Porter.
\newblock {Multilayer Networks}.
\newblock {\em {ArXiv} e-prints}, Sept. 2013.

\bibitem{kolda2009tensor}
T.~G. Kolda and B.~W. Bader.
\newblock {Tensor decompositions and applications}.
\newblock {\em SIAM review}, 51(3):455--500, 2009.

\bibitem{lambiotte_laplacian_2008}
R.~Lambiotte, J.-C. Delvenne, and M.~Barahona.
\newblock {Laplacian dynamics and multiscale modular structure in networks}.
\newblock {\em {arXiv} preprint {arXiv}:0812.1770}, 2008.

\bibitem{lambiotte_flow_2011}
R.~Lambiotte, R.~Sinatra, J.-C. Delvenne, T.~S. Evans, M.~Barahona, and
  V.~Latora.
\newblock {Flow graphs: Interweaving dynamics and structure}.
\newblock {\em {\textbackslash}pre}, 84(1):017102, July 2011.

\bibitem{Michoel2012spectral}
T.~{Michoel} and B.~{Nachtergaele}.
\newblock {Alignment and integration of complex networks by hypergraph-based
  spectral clustering}.
\newblock {\em Physical Review E}, 86(5):056111, Nov. 2012.

\bibitem{mucha_community_2010}
P.~J. Mucha, T.~Richardson, K.~Macon, M.~A. Porter, and J.-P. Onnela.
\newblock {Community {Structure} in {Time}-{Dependent}, {Multiscale}, and
  {Multiplex} {Networks}}.
\newblock {\em Science}, 328:876--, May 2010.

\bibitem{newman2010networks}
M.~Newman.
\newblock {\em {Networks: An Introduction}}.
\newblock OUP Oxford, 2010.

\bibitem{Page99thepagerank}
L.~Page, S.~Brin, R.~Motwani, and T.~Winograd.
\newblock {The PageRank Citation Ranking: Bringing Order to the Web}, 1999.

\bibitem{sole-ribalta_centrality_2014}
A.~Sol{\'e}-Ribalta, M.~{De Domenico}, S.~G{\'o}mez, and A.~Arenas.
\newblock {Centrality {Rankings} in {Multiplex} {Networks}}.
\newblock In {\em {Proceedings of the 2014 {ACM} {Conference} on {Web}
  {Science}}}, {{WebSci} '14}, pages 149--155, New York, NY, USA, 2014. ACM.

\bibitem{taylor_eigenvector-based_2015}
D.~Taylor, S.~A. Myers, A.~Clauset, M.~A. Porter, and P.~J. Mucha.
\newblock {Eigenvector-{Based} {Centrality} {Measures} for {Temporal}
  {Networks}}.
\newblock {\em arXiv preprint arXiv:1507.01266}, 2015.

\bibitem{verbrugge1979multiplexity}
L.~M. Verbrugge.
\newblock {Multiplexity in adult friendships}.
\newblock {\em Social Forces}, 57(4):1286--1309, 1979.

\bibitem{wasserman1994social}
S.~Wasserman and K.~Faust.
\newblock {\em {Social network analysis: Methods and applications}}, volume~8.
\newblock Cambridge university press, 1994.

\end{thebibliography}

\end{document}